\documentclass{article}

\PassOptionsToPackage{numbers,compress}{natbib}

\usepackage[preprint]{neurips_2026}

\usepackage[utf8]{inputenc} %
\usepackage[T1]{fontenc}    %
\usepackage{hyperref}       %
\usepackage{url}            %
\usepackage{booktabs}       %
\usepackage{amsfonts}       %
\usepackage{nicefrac}       %
\usepackage{microtype}      %
\usepackage{xcolor}         %

\usepackage{amsmath}
\usepackage{amssymb}
\usepackage{mathtools}
\usepackage{amsthm}
\usepackage{algorithm}
\usepackage{algorithmic}

\usepackage{graphicx}
\usepackage{subcaption}
\usepackage{bm}
\usepackage{multirow}

\usepackage{adjustbox}

\theoremstyle{plain}
\newtheorem{theorem}{Theorem}[section]
\newtheorem{proposition}[theorem]{Proposition}
\newtheorem{lemma}[theorem]{Lemma}

\theoremstyle{definition}

\theoremstyle{remark}
\newtheorem{remark}[theorem]{Remark}

\title{Quantum End-to-End Learning for Contextual Combinatorial Optimization}

\author{%
    Jaehwan Lee$^{1}$, Changhyun Kwon$^{1,2}$\\
    $^{1}$KAIST\quad$^{2}$Omelet\\
}

\begin{document}

\maketitle

\begin{abstract}
 Contextual combinatorial optimization (CCO) plays a critical role in decision-making under uncertainty, yet remains a significant challenge.
 We present Quantum End-to-End Learning (QEL), the first quantum computing-based end-to-end learning framework for CCO that leverages Quantum Approximate Optimization Algorithms.
 Inspired by the integration of state preparation and evolution in data re-uploading, we propose a context re-uploading phase-separator that jointly captures the complex relations among contexts, uncertain coefficients, and optimal solutions.
 This allows a contextual encoder to be seamlessly integrated within a quantum surrogate policy, enabling joint end-to-end training with a stationarity guarantee.
 Exploiting an optimization-aware structure grounded in physical principles that classical methods cannot readily leverage, our approach demonstrates practicality by directly training on task loss despite the discreteness and nonconvexity, while avoiding calls to NP-hard optimization solvers.
 QEL empirically achieves competitive performance while requiring substantially fewer parameters than classical benchmarks, highlighting its industrial-level potential for the future quantum era.
\end{abstract}

\section{Introduction}
From energy-efficient supply chain to portfolio management, numerous real-world decision-making problems require accounting for uncertainty and discreteness, leading to the need to find optimal policies given their combinatorial nature \cite{hentenryck2006online, juan2015review}.
With the rapid advances in machine learning (ML) techniques, decision-makers can exploit contextual information within data by quantifying uncertainties inherent in problem coefficients.
To this end, a \emph{contextual optimization} framework has been widely developed within the paradigm of integrating ML and operations research \cite{lepenioti2020prescriptive, sadana2025survey}.
We provide a mathematical formulation for contextual optimization in Appendix \ref{apd:contextual}.
Specifically, \emph{end-to-end learning} \cite{donti2017task} has gained significant attention as it trains ML models to directly minimize the \emph{task loss}, the downstream costs incurred by the decision \cite{sadana2025survey}.
This class of approaches is categorized into \emph{decision rule} (DR) and \emph{Predict-and-Optimize} (PnO).

In DR \citep{ban2019big}, the optimal policy is directly approximated by training an ML model that solves the empirical risk minimization (ERM) problem without the requirement of true optimal solutions:
\begin{equation}
    \label{eq:erm}
    \min_{\bm{w} \in \mathbb{R}^{d_{\bm{w}}} }\frac{1}{N}\sum_{l=1}^N C\bigl(\pi_{\bm{w}}(\bm{x}^{(l)}), \bm{y}^{(l)}\bigr),
\end{equation}
where $\mathcal{D}=\{ (\bm{x}^{(l)},\bm{y}^{(l)}) \}^N_{l=1}$ denotes realized training data of contexts $\bm{x}^{(l)}\in\mathcal{X}$ and uncertain coefficients $\bm{y}^{(l)}\in\mathcal{Y}$, and $\pi_{\bm{w}}$ represents a parametrized surrogate policy.
However, due to its black-box, solver-free nature, it suffers from limited interpretability and a lack of principled guidelines for selecting appropriate models.
On the other hand, PnO approaches---also called decision-focused learning \cite{mandi2024decision}---explicitly separate prediction and optimization layers, while training the model in the prediction layer end-to-end via the task loss through the optimization layer.
The training in the PnO framework can be represented as
\begin{equation}
    \label{eq:pao}
    \min_{\bm{w} \in \mathbb{R}^{d_{\bm{w}}} }\frac{1}{N}\sum_{l=1}^N C\bigl(a^*(\bm{x}^{(l)}, h_{\bm{w}}), \bm{y}^{(l)}\bigr),
\end{equation}
where $a^*(\bm{x}^{(l)},h_{\bm{w}})$ represents the minimizer of the expected cost given $\bm{x}^{(l)}$ with respect to the contextual predictor $h_{\bm{w}}$.
The key challenge in PnO is to preserve differentiability through the optimization layer back to the predictive ML model.
Prior studies have employed roundabout techniques, such as implicit differentiation \cite{amos2017optnet} for convex quadratic programs and regret-based surrogate loss \citep{elmachtoub2022smart} for linear programs.

Nevertheless, end-to-end learning for contextual combinatorial optimization (CCO) is severely hindered by the \emph{discreteness} of decision variables, which disrupts gradient calculation \citep{geng2024benchmarking}.
Prior work has explored continuous relaxations of the optimization layer \citep{agrawal2019differentiable, wilder2019melding, mandi2020smart}, but these methods are restricted to linear or convex objectives.
Recent studies in PnO frameworks have attempted to extend to general objective settings by employing learning-based surrogate costs \citep{shah2022decision, ferber2023surco} or learning-to-rank surrogate losses \citep{mandi2022decision}.
However, these approaches circumvent direct optimization of the task loss by employing surrogate formulations and require supervision over a precomputed or iteratively updated optimal-solution cache for NP-hard problems, precluding on-the-fly training.
In addition, the PnO framework strictly requires iterative, NP-hard optimization based on predictions, thereby obstructing practical industrial-level utilization with huge-sized problems.

In this work, to overcome these challenges in CCO, we present a novel end-to-end learning algorithm called \textit{Quantum End-to-End Learning} (QEL), based on quantum computing (QC), that can minimize task loss directly by gradient-based training.
Despite the discreteness and nonconvexity of the problem, QEL can directly solve linear and quadratic CCO problems and can be extended to higher-order polynomial problems, even without access to the true optimal solutions or continuous relaxation.

To the best of our knowledge, no prior work has attempted to solve CCO problems using QC, although research has explored the potential of quantum computing to achieve computational speedups over classical approaches in specific combinatorial optimization \cite{shaydulin2024evidence} and machine learning \cite{liu2021rigorous} tasks.
This is due to a disconnect between state preparation for quantum combinatorial optimization algorithms and fixed data encoding used in quantum machine learning, which has precluded leveraging QC for data-driven, context-aware combinatorial decision-making.

To overcome, we integrate two quantum paradigms: Quantum Approximate Optimization Algorithms (QAOA) \citep{farhi2014quantum} and data re-uploading \cite{perez2020data}, thereby harnessing the advantages of both within a unified framework.
Specifically, the theoretical foundations of QAOA permit training QEL to minimize empirical risk directly.
Building upon data re-uploading that unifies data encoding and parametrized computation, in QEL, the contextual encoder is embedded within the surrogate policy, enabling context-aware decision-making within the integration of quantum machine learning and quantum optimization.
We provide the convergence-to-stationarity guarantee of policy-encoder joint learning, demonstrating its end-to-end nature.
In addition, compared to classical approaches, it enables an optimization-aware structure with a high inductive bias, improving interpretability, eliminating the need for optimization solvers at test time, and substantially reducing the number of parameters.

\paragraph{Contribution}
Our contributions are summarized below:
\begin{itemize}
    \item We introduce the first quantum computing framework for end-to-end learning to solve contextual combinatorial optimization by integrating QAOA with data re-uploading, enabling context-aware decision-making that was previously unattainable in QC.
    \item We establish the theoretical connection between the quantum loss function derived from the variational principle and the empirical risk minimization in contextual optimization, enabling training directly on task loss, even with the discreteness of the decision variables.
    \item Our framework integrates the contextual encoder directly into a principled quantum surrogate policy, thereby strengthening it with advantageous inductive bias to improve interpretability and reduce the dimensionality of classical control spaces.
    \item We empirically validate the competitiveness of our approach, achieving decision qualities comparable to classical benchmarks.
\end{itemize}

\section{Preliminaries} \label{sec:preliminaries}

This section introduces two principles in quantum mechanics that underpin quantum-classical hybrid algorithms and quantum combinatorial optimization, providing the theoretical background of QEL for CCO.

\subsection{The Variational Principle}
The variational principle provides theoretical justification for quantum-classical hybrid algorithms, known as Variational Quantum Algorithms (VQA) \cite{mcclean2016theory, cerezo2021variational}, which use classical optimizers to train a parametrized quantum circuit (PQC).
The PQC consists of a sequence of quantum gates with tunable rotation angles, enabling the preparation of quantum states parametrized by classical variables.
The hybrid optimization loop operates by executing the PQC on a quantum processor to estimate expectation values, typically constituting the computational bottleneck, and then using a classical optimizer to update the circuit parameters iteratively, as illustrated in Figure \ref{fig:pqc}.
\begin{figure}
    \centering
    \includegraphics[scale=0.275]{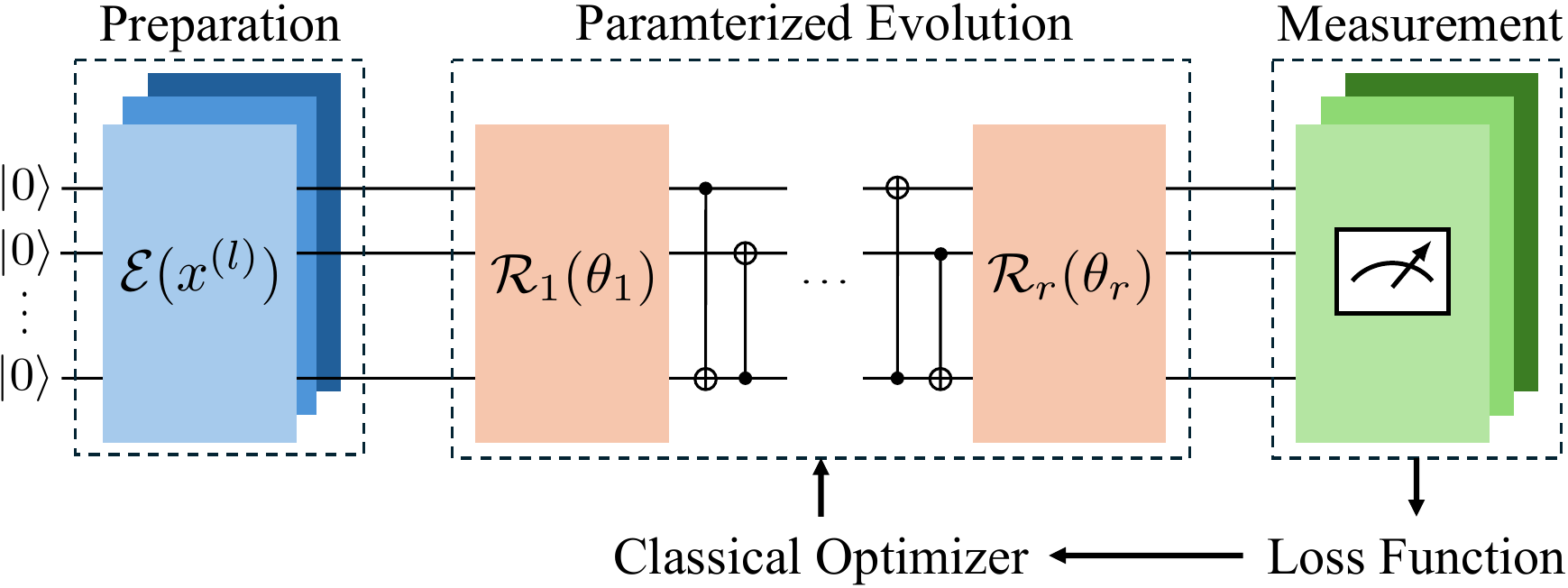}
    \caption{Schematic of a parametrized quantum circuit (PQC) and the hybrid algorithm loop. $\mathcal{E}$ denotes the data encoding unitary, and $\mathcal{R}_1,\cdots,\mathcal{R}_r$ denotes parametrized rotation unitaries.}
    \label{fig:pqc}
\end{figure}
The variational principle provides a \emph{lower bound} on the expected value of a particular Hamiltonian $\mathcal{H}$ with respect to a trial wave function called the \emph{ansatz} $|\psi\rangle$.
\begin{theorem}[The Variational Principle, \citealt{griffiths2018introduction}]
  \label{thm:variational}
  Let $\mathcal{H}$ be a Hamiltonian with ground state energy, i.e., the smallest eigenvalue, $E_0$.
  For any normalized ansatz $|\psi\rangle$,
  \begin{equation}
    \label{eq:variational}
    E_0 \le \langle\psi|\mathcal{H}|\psi\rangle,
  \end{equation}
  where $\langle\psi|\mathcal{H}|\psi\rangle$ denotes the expected value of the Hamiltonian $\mathcal{H}$ with respect to the ansatz $|\psi\rangle$.
\end{theorem}
All qubit states are normalized, and the ansatz can be implemented as a fixed PQC with rotation angle parameters $\bm{\theta}$, represented as $|\psi(\bm{\theta})\rangle$.
The training objective of VQA is to find $\bm{\theta}$ that minimizes $\langle\psi(\bm\theta)|\mathcal{H}|\psi(\bm\theta)\rangle$.
This process is performed via iterative updates with classical optimizers.

\subsection{The Adiabatic Theorem}
The adiabatic theorem serves as a foundation for solving combinatorial optimization problems in QC.
\begin{theorem}[The Adiabatic Theorem, \citealt{farhi2000quantum}]
  \label{thm:adiabatic}
  A quantum system initialized in the ground state of a time-dependent Hamiltonian will remain in the instantaneous ground state if the Hamiltonian evolves sufficiently slowly.
\end{theorem}

It can be leveraged for quantum optimization by constructing the time-dependent Hamiltonian as
\begin{equation}
    \label{eq:tdhamiltonian}
    \mathcal{H}(t)=(1-f(t))\mathcal{H}_I+f(t)\mathcal{H}_F,
\end{equation}
where $f(t)$ is a smooth and monotonic function of $t\in[0, T]$ satisfying $f(0)=0$ and $f(T)=1$, $\mathcal{H}_I$ is an initial Hamiltonian whose ground state is easily preparable, and $\mathcal{H}_F$ is a final Hamiltonian we want to solve.
Typically, a transverse-field Hamiltonian and its ground state $|\bm{s}\rangle$ is prepared as
\begin{equation}
    \label{eq:initial}
    \mathcal{H}_I=-\sum_{i=1}^{n}X_i,\quad|\bm{s}\rangle=\frac{1}{\sqrt{2^n}}\sum_{\bm{z}}|\bm{z}\rangle,
\end{equation}
where $n$ is the total number of qubits, $X_i$ is the Pauli-$X$ operator acting on the $i$th qubit, and $|\bm{z}\rangle$ is one of the computational basis states that encodes an $n$-bit string $\bm{z}$, to be initialized.
Thus, the smallest eigenvalue corresponds to the equal superposition of all $n$-bit strings.

$\mathcal{H}_F$ should encode an objective function $f$ to be minimized through
\begin{equation}
    \label{eq:final}
    \mathcal{H}_F|\bm{z}\rangle= f(\bm{z})|\bm{z}\rangle,
\end{equation}
meaning that it is required to act diagonally on $n$-qubit computational basis states, i.e.,
\begin{equation}
    \label{eq:diagonal}
    \mathcal{H}_F=\sum_{\bm{z}}f(\bm{z})|\bm{z}\rangle\langle\bm{z}|.
\end{equation}

Thus, to avoid many-body interactions that cannot be directly modeled by one- and two-qubit gates compatible with current quantum hardware \cite{chancellor2017circuit}, a typical choice is the Ising Hamiltonian \citep{lucas2014ising}
\begin{equation}
    \label{eq:ising}
    \mathcal{H}_F=\sum_{i\in V}h_iZ_i+\sum_{(i,j)\in E}y_{ij}Z_iZ_j,
\end{equation}
where $Z_i$ is the Pauli-$Z$ operator acting on the $i$th qubit.
It maps bijectively to the objective of quadratic unconstrained binary optimization (QUBO) via the substitution $Z_i=1-2x_i$, where $x_i\in\{0,1\}$ denotes the binary decision variable.
Thus, any combinatorial problem reducible to QUBO can be transformed into \eqref{eq:ising} and solved via Theorem \ref{thm:adiabatic}.
However, finding the ground state energy in \eqref{eq:ising} is NP-hard in general \cite{barahona1982computational}, forcing the required time evolution $T$ to increase exponentially in the number of qubits $n$ \cite{albash2018adiabatic}. In practice, it is expected that the ground state of $\mathcal{H}_F$ will be prepared with some high probability despite the relaxation of the time requirement \cite{yarkoni2022quantum}.

\section{Methodology}\label{sec:methodology}
This section presents QEL, our end-to-end learning framework for CCO.
We begin by reviewing QAOA as an approximate implementation of adiabatic evolution, then propose a context re-uploading phase-separator that integrates context-aware recognition of uncertainties into the quantum optimization process.
We conclude by formalizing the variational principle-based training objective and its connection to empirical risk minimization in contextual optimization, and demonstrate that candidate solutions can be extracted without access to the true uncertain coefficients for test-time inference.
These enable QEL to be directly trained via gradients of the task loss, unlike classical methods.

\subsection{Quantum Approximate Optimization Algorithm}
\label{sec:qaoa}
QAOA, proposed by \citet{farhi2014quantum}, belongs to the class of VQA and provides an approximate implementation of Theorem \ref{thm:adiabatic}.
In QC, unitary operators are the fundamental units of computation.
Therefore, \eqref{eq:tdhamiltonian} must be expressed as a matrix exponential, called a time-ordered exponential,
\begin{equation}
    \label{eq:tdunitary}
    \mathcal{U}(t)=\mathcal{T}\exp\left[-i\int^t_0 \mathop{}\!\mathrm{d}\tau\mathcal{H}(\tau)\right],
\end{equation}
where $\mathcal{T}$ is the time-ordering operator.
However, due to the non-commutativity between \eqref{eq:initial} and \eqref{eq:ising}, i.e., $[\mathcal{H}_I,\mathcal{H}_F]\ne 0$, \eqref{eq:tdunitary} does not admit a closed-form solution, necessitating approximations such as Trotterization \cite{lloyd1996universal},
\begin{align*}
    \mathcal{U}(t)&\approx\prod_{k=1}^{p}\exp[-i\mathcal{H}(k\Delta\tau)\Delta\tau]\\
    &=\prod_{k=1}^{p}\exp[-ig(k\Delta\tau)\Delta\tau\mathcal{H}_I]\exp[-i\tilde{g}(k\Delta\tau)\Delta\tau\mathcal{H}_F],
\end{align*}
where $g(\cdot)=1-\tilde{g}(\cdot)$ and $\Delta\tau=\frac{t}{p}$.

QAOA circuits are obtained by replacing the time-dependent coefficients with trainable parameters,
\begin{equation}
    \label{eq:qaoa}
    \mathcal{U}(\bm{\theta}_I,\bm{\theta}_F)=\prod_{k=1}^{p}\exp[-i\theta^k_I\mathcal{H}_I]\exp[-i\theta^k_F\mathcal{H}_F].
\end{equation}
We refer $\exp[-i\theta^k_I\mathcal{H}_I]$ and $\exp[-i\theta^k_F\mathcal{H}_F]$ as \emph{mixer} and \emph{phase-separator}, respectively.
The phase-separator separates relative phases of quantum states to distinguish near-optimal from suboptimal solutions.
Through alternation with the mixer operator, this mechanism induces quantum interference that preferentially amplifies the probability of sampling near-optimal solutions.
In terms of VQA, the alternation of the phase-separator and mixer repeated $p$ times constitutes the ansatz.
This ansatz is shown to converge to adiabatic evolution as $p\to\infty$ \citep{farhi2014quantum}.
However, due to hardware immaturity, most studies on the practical use of QAOA have employed $p\le5$ \cite{zhou2020quantum, blekos2024review}.
In addition, we can extend the parametrization of QAOA by decomposing the phase-separator and mixer (see Appendix \ref{apd:parametrization}).

QAOA has been widely regarded as a traditional heuristic algorithm for solving a single deterministic QUBO problem.
However, recent efforts have extended its applicability beyond quadratic formulations by quadratization techniques or gate decomposition \citep{glos2022space, pelofske2024short}.
In addition, computational universality and extensibility of QAOA have been established \citep{lloyd2018quantum, morales2020universality}.
Specifically, \citet{lloyd2018quantum} proposed the application of QAOA to quantum machine learning (QML) to implement a unitary transformation that maps inputs to outputs, given a training set of input-output pairs.
In this work, we examine the generalizability of QAOA.
Particularly, by proposing the QEL framework, we focus on its applicability to solving CCO problems.

\subsection{Context Re-uploading Phase-Separator}
In this section, we introduce a \emph{context re-uploading phase-separator}, the core component of QEL designed to address the CCO problem.
Motivated by the philosophy of data re-uploading that integrates state preparation with evolution, we embed covariates directly within the phase-separator using trainable parameters, thereby enabling simultaneous context recognition and optimization.
The context re-uploading phase-separator learns intricate relationships between covariates and uncertain parameters via its internal contextual encoder, while simultaneously separating relative phases between near-optimal and suboptimal solutions.

Conventional QML algorithms employ a fixed data encoding scheme as a state preparation step, such as amplitude encoding \cite{lloyd2013quantum} and kernel-based encoding \cite{havlivcek2019supervised}, to map classical data onto quantum states, which are then processed through PQCs in the evolution step.
However, QAOA strictly requires the system to be initialized in a particular state, as in \eqref{eq:initial}.
This methodological gap presents a significant barrier to developing contextual optimization approaches grounded in QML and quantum algorithms for combinatorial optimization.
Alternatively, an emerging paradigm of data re-uploading integrates preparation and evolution by repeatedly uploading data with trainable linear parameters, and its improved trainability and expressivity through interleaving have been studied \cite{perez2020data, perez2021one, schuld2021effect}.

Motivated by original data re-uploading \cite{perez2020data}, we can construct a context-dependent final Hamiltonian by using a linear encoder to quantify uncertain coefficients $y_{ij}$ in \eqref{eq:ising} from covariates $\bm{x}_{ij}$,
\begin{equation}
    \hat{\mathcal{H}}_F=\sum_{i\in V}h_iZ_i+\sum_{(i,j)\in E}(w_0+\bm{w}_1^\top \bm{x}_{ij})Z_iZ_j.
\end{equation}
Replacing the original final Hamiltonian with $\hat{\mathcal{H}}_F$, we can upload the classical data while evolving the quantum state in the phase-separator,
\begin{align}
    \label{eq:dups_linear}
    \exp[-i\theta^k_F\mathcal{H}_F] \approx\exp[-i\theta^k_F\hat{\mathcal{H}}_F] .
\end{align}
In QEL, the context re-uploading phase separator is designed to search near-optimal solution space accelerated by the theory of QAOA, while general-purpose data re-uploading requires a decomposition-based representation to guarantee global search.

We compute stochastic gradients for PQCs using quantum-native algorithms such as the parameter-shift rule \cite{schuld2019evaluating, wierichs2022general}.
Following \cite{cervera2021meta}, we extend it to more general encoding strategies,
\begin{equation}
    \label{eq:dups_general}
    y_{ij}\approx\hat{y}_{ij}= g_{\bm{w}}(\bm{x}_{ij})\quad\forall{(i,j)\in E},
\end{equation}
where $g_{\bm{w}}(\bm{x}_{ij})$ denotes a general encoder and $\bm{w}$ represents its parameter set.
Specifically, by exploiting the bounded nature of rotation angle parameters in PQCs, we propose a logistic encoder not explored in the original data re-uploading to compare to linear encoders,
\begin{equation}
    \label{eq:dups_logistic}
    g_{\bm{w}}(\bm{x}_{ij})=\frac{w_0}{1+\exp[-\bm{w}_1^\top (\bm{x}_{ij}-\bm{w}_2)]}\quad\forall{(i,j)\in E},
\end{equation}
where $\bm{w}=\{w_0,\bm{w}_1,\bm{w}_2\}$.
Furthermore, data re-uploading is achieved by increasing the number of alternations $p$. 
The schematic structure of the context re-uploading phase-separator and mixer is given in Figure \ref{fig:QEL}.
\begin{figure}
    \centering
    \includegraphics[scale=0.25]{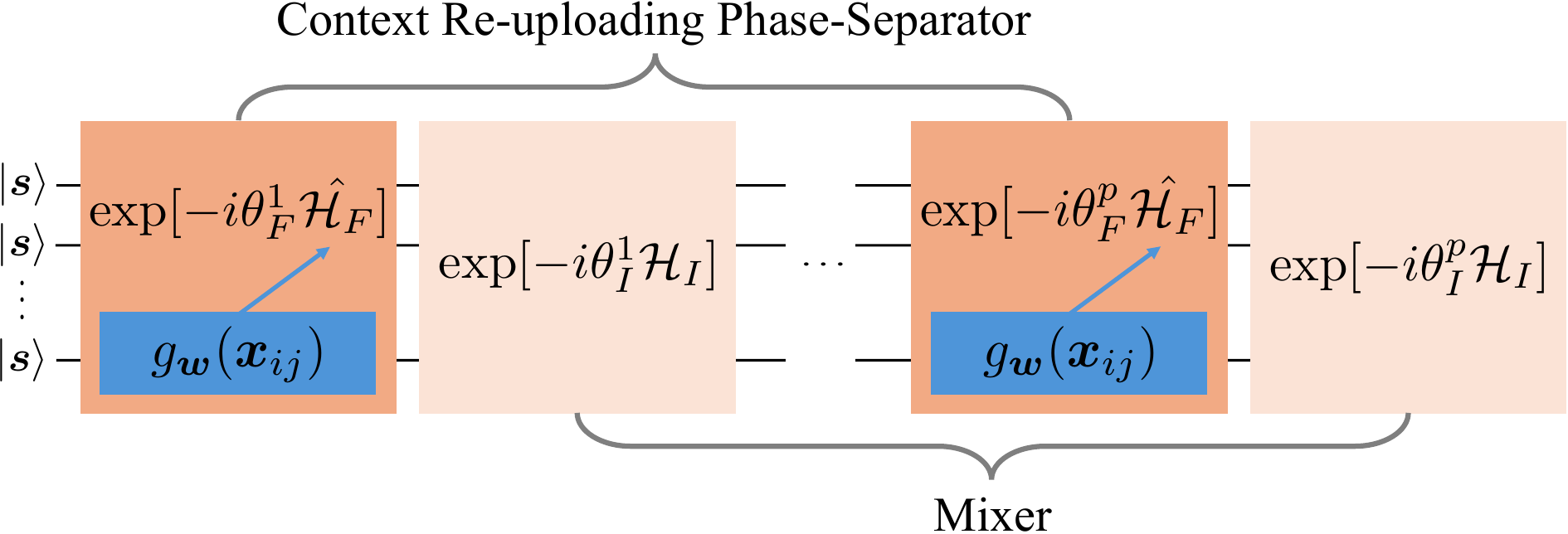}
    \caption{Architecture of the QEL ansatz extending context re-uploading phase-separators. The contextual encoder $g_{\bm{w}}(\bm{x})$ repeatedly uploads latent information of contexts directly into the context-dependent final Hamiltonian $\hat{\mathcal{H}}_F$ in the phase-separator.}
    \label{fig:QEL}
\end{figure}

Note that the optimization-aware structure based on QAOA allows QEL to have a high inductive bias, thereby improving the interpretability of its performance and enabling principled decision-making without iterative calls to optimization solvers.
In addition, it substantially reduces the number of parameters, reducing the dimensionality of the classical control.
In the future quantum era, with scalable hardware and advanced error correction, the need to tune too many classical control variables may reduce the efficiency advantages of quantum operations, offering the possibility for QEL to solve industrial-level problems efficiently.

\subsection{Training and Test-time Processes}
\label{sec:processes}
In this section, we demonstrate the trainability and testability of QEL.
Specifically, we connect the minimization of QEL loss, following the variational principle, to the empirical loss minimization of the task loss in end-to-end frameworks.
In addition, we provide a stationarity guarantee of joint policy-encoder optimization.
Then, we demonstrate how to make test-time decisions using only contextual information, without true coefficients.
We provide pseudo code for the training and test-time processes in Appendix \ref{apd:pseudocode}.

\paragraph{Training Process}
Given that QEL is an extension of QAOA, we can define its loss function following Theorem \ref{thm:variational}.
Suppose that QEL follows the original QAOA parametrization strategy.
Let $N$ be the number of i.i.d. historical data pairs $\{(\bm{x}^{(l)},\bm{y}^{(l)})\}_{l=1}^N$ where $\bm{x}^{(l)}=(\bm{x}^{(l)}_{ij}:{(i,j)\in E})$ and $\bm{y}^{(l)}=(y_{ij}^{(l)}: {(i,j)\in E})$.
The final Hamiltonian corresponding to $\bm{y}^{(l)}$ is referred to as $\mathcal{H}_F(\bm{y}^{(l)})$.
We denote the context-dependent final Hamiltonian for $\bm{y}^{(l)}$ by the encoder $g_{\bm{w}}$ as $\hat{\mathcal{H}}_F\bigl(g_{\bm{w}}(\bm{x}^{(l)})\bigr)$.
Then the QEL ansatz with $p$ alternations can be represented as
\begin{equation}
    \label{eq:QEL_ansatz}
    \mathcal{U}(\bm{\phi},\bm{x}^{(l)})
    :=\prod_{k=1}^p\exp[-i\theta^k_I\mathcal{H}_I]\exp[-i\theta^k_F\hat{\mathcal{H}}_F\bigl(g_{\bm{w}}(\bm{x}^{(l)})\bigr)],
\end{equation}
where $\bm{\phi}=\{\bm\theta_I,\bm\theta_F,\bm{w}\}$ is the entire parameter set of QEL.
Thus, the final state output by the circuit is given by
\begin{equation}
    \label{eq:QEL_final}
    |\psi(\bm{\phi},\bm{x}^{(l)})\rangle := \mathcal{U}(\bm{\phi},\bm{x}^{(l)})|\bm{s}\rangle,
\end{equation}
where $|\bm{s}\rangle$ is defined in \eqref{eq:initial}.
    QEL is optimized over a set of historical data pairs by minimizing a cost function that depends on all expected values of the Hamiltonians $\mathcal{H}_F(\bm{y}^{(l)})$ with respect to the final state \eqref{eq:QEL_final}.
    Therefore, the loss function is defined as an average of expected values:
\begin{equation}
    \label{eq:QEL_loss}
    L(\bm{\phi}):=\frac{1}{N}\sum_{l=1}^N\langle\psi(\bm{\phi},\bm{x}^{(l)})|\mathcal{H}_F(\bm{y}^{(l)})|\psi(\bm{\phi},\bm{x}^{(l)})\rangle,
\end{equation}
and the parameters within the QEL ansatz are trained to minimize the loss function by classical optimizers.
Note that the QEL ansatz construction in \eqref{eq:QEL_ansatz} could be extended by adopting more sophisticated QAOA parametrization and encoding strategies.

We now establish the connection between the QEL loss \eqref{eq:QEL_loss} based on the variational principle and empirical risk minimization for contextual optimization \eqref{eq:erm}.
\begin{theorem}[QEL as End-to-End Learning]
    \label{prop:QEL_erm}
    Let $\pi_{\bm{\theta}}\bigl(\bm{x},g_{\bm{w}}\bigr)$ denote the stochastic policy induced by measuring the final state $|\psi({\bm{\phi}},\bm{x})\rangle$ in the computational basis, where $g_{\bm{w}}$ is the contextual encoder.
    Define the cost function $C[\pi, \bm{y}] := \mathbb{E}_{\bm{z}\sim\pi}[f(\bm{z}, \bm{y})]$ where $f(\bm{z}, \bm{y}) := \langle \bm{z}|\mathcal{H}_F(\bm{y})|\bm{z}\rangle$.
    Then minimizing \eqref{eq:QEL_loss} is equivalent to
    \begin{equation}
        \label{eq:QEL_erm}
        \min_{\bm{\phi}}
        \frac{1}{N}
        \sum_{l=1}^N C \Bigl[
        	\pi_{{\bm{\theta}}} \bigl(
				\bm{x}^{(l)},g_{\bm{w}}
			\bigr), \bm{y}^{(l)}
		\Bigr].
    \end{equation}
\end{theorem}
\begin{proof}
    Since $\mathcal{H}_F(\bm{y}^{(l)})$ is diagonal in the computational basis, we have $\mathcal{H}_F(\bm{y}^{(l)}) = \sum_{\bm{z}} f(\bm{z}, \bm{y}^{(l)}) |\bm{z}\rangle\langle\bm{z}|$.
    Expanding the expectation value, we obtain
    \begin{align*}
        \langle\psi({\bm{\phi}},\bm{x}^{(l)})|\mathcal{H}_F(\bm{y}^{(l)})|\psi({\bm{\phi}},\bm{x}^{(l)})\rangle
        &= \sum_{\bm{z}} f(\bm{z}, \bm{y}^{(l)})|\langle \bm{z}|\psi({\bm{\phi}},\bm{x}^{(l)})\rangle|^2 \\
        &= \mathbb{E}_{\bm{z}\sim\pi_{{\bm{\theta}}}(\bm{x}^{(l)},g_{\bm{w}})}[f(\bm{z}, \bm{y}^{(l)})] \\
        &= C\bigl[\pi_{{\bm{\theta}}}(\bm{x}^{(l)},g_{\bm{w}}),\bm{y}^{(l)}\bigr].
    \end{align*}
    Substituting this result into \eqref{eq:QEL_loss} yields \eqref{eq:QEL_erm}.
\end{proof}
Instantiated by end-to-end learning, this theorem reveals the trainability of QEL, where the contextual encoder $g_{\bm{w}}$ and the quantum surrogate policy $\pi_{{\bm{\theta}}}$ are jointly trained end-to-end via empirical risk minimization of the downstream task loss $C$ directly, despite the discreteness and nonconvexity.

In addition, a canonical diminishing-step stochastic gradient recursion of joint QEL training, using full-batch finite-measurement parameter-shift gradients \cite{schuld2019evaluating, wierichs2022general}, admits convergence to stationarity.

\begin{theorem}[Convergence to stationarity]
\label{thm:qel_diminishing_step_stationarity}
    Consider the update $\bm{\phi}_{t+1}=\bm{\phi}_t-\eta_t\bm{g}_t,$
    where \(\bm{g}_t\) is a full-batch finite-measurement parameter-shift estimator.
    Suppose the regularity conditions \textup{(R1)}--\textup{(R4)} in Appendix~\ref{app:stationary_convergence_qel} hold.
    Assume that the diminishing step-size sequence satisfies
    \[
        0<\eta_t\le \frac{1}{\beta},
        \qquad
        \eta_t\to 0,
        \qquad
        \sum_{t=0}^{\infty}\eta_t=\infty,
        \qquad
        \sum_{t=0}^{\infty}\eta_t^2<\infty,
    \]
    where $\beta$ denotes the smoothness constant, supplied by Lemma~\ref{lem:appendixF_smoothness}.
    Then the iterates are asymptotically stationary in expected squared gradient norm:
    \[
        \lim_{t\to\infty}
        \mathbb E\!\left[\|\nabla L(\bm{\phi}_t)\|_2^2\right]
        =0.
    \]
    Moreover, the full stochastic trajectory is asymptotically stationary in the almost-sure sense:
    \[
        \mathbb P\!\left(
            \lim_{t\to\infty}\|\nabla L(\bm{\phi}_t)\|_2 = 0
        \right)=1.
    \]
\end{theorem}
The full regularity conditions, lemmas, and a proof are deferred to Appendix~\ref{app:stationary_convergence_qel}.

For practical quantum hardware, exact information of $|\psi(\bm{\phi}, \bm{x}^{(l)})\rangle$ is not accessible.
Alternatively, we are required to statistically estimate \eqref{eq:QEL_loss} from measurements.
Let $N_\text{shots}$ be the number of measurements called shots, and suppose that $n$-bit string samples $\{\bm{z}_1^{(l)},\bm{z}_2^{(l)},\cdots,\bm{z}^{(l)}_{N_\text{shots}}\}$ are measured by the observable $\mathcal{H}_F(\bm{y}^{(l)})$ with respect to the state \eqref{eq:QEL_final}.
Then the loss function \eqref{eq:QEL_loss} is estimated as
\begin{equation}
    \label{eq:QEL_loss_est}
    \tilde{L}({\bm{\phi}}) = \frac{1}{N}\sum_{l=1}^N \frac{1}{N_\text{shots}} \sum_{m=1}^{N_\text{shots}} f(\bm{z}_m^{(l)}, \bm{y}^{(l)}),
\end{equation}
where $\bm{z}_m^{(l)} \sim |\langle \bm{z} | \psi({\bm{\phi}}, \bm{x}^{(l)}) \rangle|^2$.
By iteratively updating the parameters to minimize \eqref{eq:QEL_loss_est}, the circuit learns to amplify the probability of sampling near-optimal solutions while suppressing low-quality ones.
The computational overhead of shot-based estimation remains manageable in practice \citep{weidenfeller2022scaling}.

\paragraph{Test-time Process}
Based on the trained stochastic policy distribution, at test time, CCO problems require concrete bitstring decisions rather than a distribution, so one must decode the learned output distribution into a single candidate solution using sampling.
For test-time inference, the true uncertain coefficients are not accessible at the decision time.
Instead, we measure all qubits of the trained circuit in the computational basis, sample $N_{\text{shots}}$ bitstrings, decode one candidate solution, and utilize it as a prescriptive decision.
This procedure is justified by the following equivalence, established via diagonalization.

\begin{proposition}[Equivalence of Sampling]
    \label{lem:sampling_equiv}
    Let $|\psi\rangle$ be an $n$-qubit state and $\mathcal{H}_F(\bm{y})$ be the Ising Hamiltonian \eqref{eq:ising} with uncertain coefficients $\bm{y}$.
    Then sampling from the computational-basis measurement distribution and subsequently evaluating $f(\bm{z}, \bm{y})$ for each sample yields the same expectation as directly estimating $\langle\psi|\mathcal{H}_F(\bm{y})|\psi\rangle$.
\end{proposition}
\begin{proof}
    Since $\mathcal{H}_F(\bm{y})$ is diagonal in the computational basis, measuring all qubits in the computational basis yields bitstring $\bm{z}$ with probability $p(\bm{z}) = |\langle\bm{z}|\psi\rangle|^2$.
    The expectation of $f(\bm{z}, \bm{y})$ under this distribution is
    \begin{align*}
        \sum_{\bm{z}} p(\bm{z}) f(\bm{z}, \bm{y})&=\sum_{\bm{z}} |\langle\bm{z}|\psi\rangle|^2 \langle\bm{z}|\mathcal{H}_F(\bm{y})|\bm{z}\rangle= \langle\psi|\mathcal{H}_F(\bm{y})|\psi\rangle,
    \end{align*}
    which completes the proof.
\end{proof}
Therefore, from the trained policy distribution of QEL, we can sample candidate solutions at test time without knowledge of $\bm{y}$.
In standard QAOA workflows, the most common way to decode is to use the most probable bitstring \citep{sharma2022openqaoa}, while advanced decoding strategies could exploit this distribution to obtain higher-quality and more refined candidate solutions.

\section{Experimental Evaluation}
To evaluate the potential of QEL, we conduct experiments on the GPU-based quantum simulator and compare the decision performance with classical benchmarks.
Detailed simulation hardware/software specifications, data generation, and hyperparameters are provided in Appendix \ref{apd:implementation}.

\subsection{Experimental Setup}
\label{sec:experimental}
\paragraph{Problems}
We focus on two quadratic problems, namely the undirected weighted maximum cut problem (MaxCut) and the quadratic assignment problem (QAP), as well as the bipartite matching problem (BMP) in \citep{elmachtoub2022smart}.
Nevertheless, co-using the techniques presented in Section \ref{sec:qaoa} suggests the potential for generalization of QEL to higher-order formulations.
Due to computational constraints in our simulation environment, we restrict our experimental analysis to problem instances with 16 and 25 variables for MaxCut and QAP, and 25 for BMP.
These scale limitations are anticipated to be overcome as quantum technology advances and becomes more accessible.

We employ constraint penalization to solve QAP and BMP \cite{lucas2014ising} and use a decoding strategy to select the most probable feasible bitstring.
While it is well-established that penalization increases the complexity of the cost landscape \cite{brandhofer2022benchmarking}, constraint-preserving mixers that restrict solution search to the feasible subspace have gained significant attention as an alternative approach \cite{hadfield2022analytical}. 
Nevertheless, we adopt the conventional penalization approach, as our primary contribution is the first QC-based framework for addressing contextual combinatorial optimization problems.
We note that it is possible to construct hybrid architectures combining context re-uploading phase-separators with constraint-preserving mixers for future extensions.
Details on the problem formulations and penalization methods for QUBO conversion are provided in Appendix \ref{apd:problem}.

\paragraph{Benchmark}
Due to the quadraticity and nonconvexity of MaxCut and QAP, benchmarks are required to handle general objective functions.
Referring to \citet{geng2024benchmarking}, we comprehensively compare QEL with Listwise loss-based Learning-to-Rank (List-LTR) in \citet{mandi2022decision} and Locally Optimized Decision Loss (LODL) approaches in \citet{shah2022decision}.
Note that List-LTR requires computing true optimal solutions for all training instances during initialization, and additionally incurs probabilistic solver calls during training to expand the solution cache, imposing a computational burden on large-scale problems.
LODL uses optimization solver calls to evaluate the decision loss of sampled perturbations around every training instance’s ground-truth label to learn surrogate losses.

\paragraph{Model Structure}
We use two contextual encoders, the linear encoder in \eqref{eq:dups_linear} and the logistic encoder in \eqref{eq:dups_logistic}, called QEL-Lin and QEL-Log, respectively, to check the influence of encoders.
We adapt the with-bias QAOA parametrization strategy to construct the surrogate policy; see Appendix \ref{apd:parametrization}.
In addition, we use $p=3,4$ in \eqref{eq:QEL_ansatz}.
For List-LTR and LODL, according to \citet{geng2024benchmarking}, we use an MLP model with the architecture $(2,32)$, where each entry denotes the number of hidden layers and the dimension of each hidden layer, respectively.

\subsection{Results}

\begin{table*}
\centering
\caption{Comparison of average relative regret and the number of parameters.}
\label{tab:result}
\begin{adjustbox}{max width=\textwidth}
\begin{tabular}{cclrrrrrrrrrrr}
\hline
                         &          &  & \multicolumn{5}{c}{Regret (\%)}                                                                                                 & \multicolumn{1}{c}{} & \multicolumn{5}{c}{\#Parameters}                                                                                                \\ \cline{4-8} \cline{10-14} 
                         &          &  & \multicolumn{2}{c}{MaxCut}                        & \multicolumn{2}{c}{QAP}                           & \multicolumn{1}{c}{BMP} & \multicolumn{1}{c}{} & \multicolumn{2}{c}{MaxCut}                        & \multicolumn{2}{c}{QAP}                           & \multicolumn{1}{c}{BMP} \\ \cline{4-8} \cline{10-14} 
                         &          &  & \multicolumn{1}{c}{16V} & \multicolumn{1}{c}{25V} & \multicolumn{1}{c}{16V} & \multicolumn{1}{c}{25V} & \multicolumn{1}{c}{25V} & \multicolumn{1}{c}{} & \multicolumn{1}{c}{16V} & \multicolumn{1}{c}{25V} & \multicolumn{1}{c}{16V} & \multicolumn{1}{c}{25V} & \multicolumn{1}{c}{25V} \\ \hline
LODL                     & $(2,32)$ &  & 9.24                    & 8.88                    & \textbf{4.03}           & \textbf{5.05}           & 80.00                   &                      & 12,728                  & 30,188                  & 2,640                   & 3,513                   & 9,313                   \\ \hline
List-LTR                 & $(2,32)$ &  & 9.34                    & 8.30                    & 4.98                    & 5.47                    & 87.20                   &                      & 12,728                  & 30,188                  & 2,640                   & 3,513                   & 9,313                   \\ \hline
\multirow{2}{*}{QEL-Lin} & $p=3$    &  & 5.56                    & 4.69                    & 4.57                    & 6.29                    & 80.00                   &                      & 9                       & 9                       & 12                      & 12                      & 266                     \\
                         & $p=4$    &  & 5.49                    & 4.97                    & 5.70                    & 5.75                    & 80.00                   &                      & 11                      & 11                      & 15                      & 15                      & 269                     \\ \hline
\multirow{2}{*}{QEL-Log} & $p=3$    &  & 5.17                    & 4.52                    & 4.55                    & 5.18                    & 75.00                   &                      & 11                      & 11                      & 14                      & 14                      & 522                     \\
                         & $p=4$    &  & \textbf{5.06}           & \textbf{4.21}           & 4.47                    & 5.60                    & \textbf{60.00}          &                      & 13                      & 13                      & 17                      & 17                      & 525                     \\ \hline
\end{tabular}
\end{adjustbox}
\end{table*}
In this section, we compare the decision quality of QEL to the benchmarks using relative regret with respect to the true optimal objective value.
For each test instance $(\bm{x}^{(l)},\bm{y}^{(l)})$, let $\hat{\bm{a}}^{(l)}$ denote the decision obtained by a given method using $\bm{x}^{(l)}$, and let $\bm{a}^{*(l)} \in \arg\min_{\bm{a}\in\mathcal{A}}
C(\bm{a},\bm{y}^{(l)})$
be the true optimal decision under the realized uncertain coefficients $\bm{y}^{(l)}$. 
We define a relative regret of $\hat{\bm{a}}^{(l)}$ as:
\[
\operatorname{RR}^{(l)} = \frac{C(\hat{\bm{a}}^{(l)},\bm{y}^{(l)})-C(\bm{a}^{*(l)},\bm{y}^{(l)})}{C(\bm{a}^{*(l)},\bm{y}^{(l)})}.
\]
We report the average relative regret on the test data, along with the number of parameters, in Table \ref{tab:result}.
In addition, to compare their decision qualities in terms of uncertainty, we provide 95\% confidence intervals of each regret result in Appendix~\ref{apd:results}.

For MaxCut, QEL consistently achieves superior decision quality compared with classical baselines, irrespective of the encoder type and the number of alternations $p$.
For BMP, despite the inherent limitations of penalty-based constraint handling, QEL-Log outperforms both classical benchmarks in terms of relative regret, while QEL-Lin achieves comparable decision quality.
In QAP, although LODL yields the best overall performance across the considered problem sizes, the best-case performance of QEL-Log still surpasses that of List-LTR.
When $p$ is fixed, QEL-Log consistently outperforms QEL-Lin for both $p=3$ and $p=4$, suggesting that incorporating a more expressive nonlinear encoder can further improve the decision quality of the QEL framework.
Moreover, except for two cases, increasing the number of alternations generally leads to lower regret under the same encoder design.
This trend is consistent with the fact that larger $p$ enables a more accurate approximation of adiabatic evolution, thereby highlighting a key strength of QEL grounded in its underlying physical principle.
Looking ahead, the emergence of scalable and error-corrected QC technology, namely fault-tolerant quantum computing (FTQC), is expected to further enhance the practical potential of QEL by enabling larger numbers of decision variables and substantially greater alternation counts.

Across the same problem instances, QEL attains comparable or superior decision quality with one to three orders of magnitude fewer trainable parameters than the classical benchmarks.
This parameter efficiency is particularly meaningful for quantum implementations, where the size of the classical control space will directly affect the overhead of quantum-classical hybrid optimization as FTQC becomes available.
This suggests that QEL may provide a resource-efficient route to practical quantum advantage while maintaining competitive decision quality.

\section{Conclusion}
In this work, we introduced Quantum End-to-End Learning (QEL), the first QC-based framework for CCO that integrates contextual encoding directly within the quantum optimization process through our proposed context re-uploading phase-separator.
By leveraging the variational and adiabatic principles, QEL enables end-to-end training on task loss without requiring access to true optimal solutions or continuous relaxation.
In addition, the optimization-aware structure of QEL improves interoperability and enables principled decision-making without heavy optimization-solver runs.
Our experiments show that QEL achieves competitive decision quality compared to classical benchmarks while requiring substantially fewer parameters, suggesting practical applicability on future quantum hardware.
Promising directions include integrating constraint-preserving mixers to guarantee feasibility, extending the framework to higher-order polynomial objectives, and validating on actual quantum processors as hardware capabilities advance.

\bibliographystyle{plainnat}
\bibliography{manuscript}

\newpage
\appendix

\section{Mathematical Formulation for Contextual Optimization}
\label{apd:contextual}
For a more detailed review of mathematical settings for contextual optimization, please refer to \citet{sadana2025survey}.
In the contextual optimization paradigm, the goal is to find a decision (i.e., an action) $\bm{a}$ within a feasible set $\mathcal{A}\subseteq\mathbb{R}^{d_{\bm{a}}}$ that minimizes a cost function $C(\bm{a},\bm{y})$, where $\bm{y}\in\mathcal{Y}\subseteq\mathbb{R}^{d_{\bm{y}}}$ represents uncertain coefficients in the objective function.
These coefficients remain unknown at decision time.
However, a vector of covariates $\bm{x}\in\mathcal{X}\subseteq\mathbb{R}^{d_{\bm{x}}}$, correlated with $\bm{y}$, is observed prior to selecting $\bm{a}$, thereby utilizing the information of contexts for decision-making.
We denote by $\mathbb{P}_{\bm{y}|\bm{x}}$ the conditional distribution of the uncertain coefficients $\bm{y}$ given the covariates $\bm{x}$.
Given the covariates $\bm{x}$, the decision-maker wants to find an optimal action minimizing the expected cost,
\begin{equation*}
    \bm{a}^*(\bm{x}) \in \underset{\bm{a} \in \mathcal{A}}{\arg\min} \, \mathbb{E}_{\mathbb{P}_{\bm{y}|\bm{x}}} [C(\bm{a}, \bm{y})]:=h(\bm{a},\mathbb{P}_{\bm{y}|\bm{x}}),
\end{equation*}
where $h(\cdot,\cdot)$ is defined as an expected cost operator that receives the action and the distribution, respectively.
By the interchangeability principle in Theorem 14.60 of \citet{rockafellar1998variational}, solving the optimal action problem for any covariate $\bm{x}$ is equivalent to solving the optimal policy to minimize the long-term expected cost called the risk,
\begin{equation*}
    \bm{\pi}^*\in\underset{\bm{\pi} \in \Pi}{\arg\min} \,\mathbb{E}_{\mathbb{P}} \bigl[C\bigl(\pi(\bm{x}), \bm{y}\bigr)\bigr],
\end{equation*}
where $\mathbb{P}$ represents the joint distribution of $\bm{x}$ and $\bm{y}$ over $\mathcal{X}\times\mathcal{Y}$ and $\Pi$ denotes the space of all feasible policies $\pi$ from $\mathcal{X}\to\mathcal{A}$.
However, the true joint distribution $\mathbb{P}$ and the true policy space $\Pi$ are unknown.
Instead, historical data $\mathcal{D}=\{ (\bm{x}^{(l)},\bm{y}^{(l)}) \}^N_{l=1}$ assumed to be independent and identically distributed (i.i.d.) realizations of $(\bm{x},\bm{y})\in\mathcal{X}\times\mathcal{Y}$ can be used as an empirical distribution, and the decision-maker aims to find a policy that approximates the optimal policy.
The data can be generalized to $\{(\bm{x}^{(l)},\hat{\mathbb{P}}^{(l)})\}$ where $\hat{\mathbb{P}}^{(l)}$ is an unbiased estimator for the conditional distribution $\mathbb{P}_{\bm{y}|\bm{x}^{(l)}}$ \citep{rychener2023end}.

\section{Parametrization Strategies of Quantum Approximate Optimization Algorithms}
\label{apd:parametrization}
We can extend the parametrization of QAOA by decomposing the phase-separator and mixer.
Since the terms in \eqref{eq:initial} and \eqref{eq:ising} commute within each Hamiltonian, we have
\begin{align}
    \exp[-i\theta_F^k\mathcal{H}_F]
    &= \exp\Biggl[
        -i\theta_F^k\biggl(
            \sum_{i\in V}h_iZ_i+\sum_{(i,j)\in E}y_{ij}Z_iZ_j
        \biggr)
    \Biggr]\notag\\
    &= \prod_{i\in V}\exp[-i\theta_F^k h_iZ_i]
    \prod_{(i,j)\in E}\exp[-i\theta_F^k y_{ij}Z_iZ_j],\notag\\
    \exp[-i\theta_I^k\mathcal{H}_I]\notag
    &= \exp\biggl[i\theta_I^k\sum_{i\in V}X_i\biggr] = \prod_{i\in V}\exp[i\theta_I^kX_i].\notag
\end{align}
The original parametrization applies a single parameter $\theta^k_F$ and $\theta^k_I$ step to all individual unitary gates of the $k$th alternation in the phase-separator and mixer, respectively.
Alternatively, following the decomposition, each unitary gate can be assigned distinct parameters, known as multi-angle QAOA \cite{herrman2022multi},
\begin{align*}
    \mathcal{U}_{F,\mathrm{MA}}^k
    &= \prod_{i\in V}\exp[-i\theta_F^{k,i}h_iZ_i]\prod_{(i,j)\in E}\exp[-i\theta_F^{k,(i,j)}y_{ij}Z_iZ_j],\\
    \mathcal{U}_{I,\mathrm{MA}}^k
    &= \prod_{i\in V}\exp[i\theta_I^{k,i}X_i],
\end{align*}
where the original parametrization corresponds to the special case $\theta_F^{k,i} = \theta_F^{k,(i,j)} = \theta_F^k\ \forall i \in V,\ (i,j) \in E$ and $\theta_I^{k,i} = \theta_I^k\quad \forall i \in V$.
While this increases expressivity, it significantly raises the classical computational cost for parameter optimization.
As a compromise between original QAOA and multi-angle QAOA, the with-bias parametrization \cite{sharma2022openqaoa} independently parametrizes the mixer, linear $Z_i$ terms, and quadratic $Z_iZ_j$ terms, while sharing parameters within each term,
\begin{align*}
    \mathcal{U}_{F,\mathrm{WB}}^k
    &= \prod_{i\in V}\exp[-i\theta_F^{k,0}h_iZ_i]\prod_{(i,j)\in E}\exp[-i\theta_F^{k,1}y_{ij}Z_iZ_j],\\
    \mathcal{U}_{I,\mathrm{WB}}^k
    &= \prod_{i\in V}\exp[i\theta_I^kX_i],
\end{align*}
which corresponds to the special case of multi-angle QAOA as $\theta^{k,i}_F = \theta_F^{k,0}\ \forall i \in V$, $\theta_F^{k,(i,j)} = \theta_F^{k,1}\ \forall (i,j) \in E$, and $\theta_I^{k,i} = \theta_I^k\quad \forall i \in V$.

\section{Pseudo Code for QEL Training and Test-time Processes}
\label{apd:pseudocode}
Algorithm~\ref{alg:QEL_training} and Algorithm~\ref{alg:QEL_test} present the training and test-time procedures of QEL, respectively.
\begin{algorithm}
  \caption{Training Process for QEL}
  \label{alg:QEL_training}
  \begin{algorithmic}
    \REQUIRE Historical data pairs $\{(\bm{x}^{(l)},\bm{y}^{(l)})\}_{l=1}^N$, number of alternations $p$, number of shots $N_{\text{shots}}$, number of epochs $N_{\text{epochs}}$
    \ENSURE Trained parameters $\bm{\phi}^*=\{\bm{\theta}_I^*,\bm{\theta}_F^*,\bm{w}^*\}$
    \STATE Initialize parameters $\bm{\phi}=\{\bm{\theta}_I,\bm{\theta}_F,\bm{w}\}$ randomly
    \STATE Initialize mixer Hamiltonian $\mathcal{H}_I$
    \FOR{$t=1$ \textbf{to} $N_{\text{epochs}}$}
    \STATE $\mathcal{L} \leftarrow 0$
    \FOR{$l=1$ \textbf{to} $N$}
    \FOR{each edge $(i,j) \in E$}
    \STATE $\hat{y}_{ij} \leftarrow g_{\bm{w}}(\bm{x}^{(l)}_{ij})$
    \ENDFOR
    \STATE $\hat{\mathcal{H}}_F \leftarrow \sum_{i\in V}h_i^{(l)}Z_i+\sum_{(i,j)\in E} \hat{y}_{ij} Z_i Z_j$
    \STATE $|\psi\rangle \leftarrow |\bm{s}\rangle$
    \FOR{$k=1$ \textbf{to} $p$}
    \STATE $|\psi\rangle \leftarrow \exp(-i\theta_F^k \hat{\mathcal{H}}_F)|\psi\rangle$
    \STATE $|\psi\rangle \leftarrow \exp(-i\theta_I^k \mathcal{H}_I)|\psi\rangle$
    \ENDFOR
    \STATE $C^{(l)} \leftarrow 0$
    \FOR{$m=1$ \textbf{to} $N_{\text{shots}}$}
    \STATE $\bm{z}_m \sim |\langle \bm{z} | \psi \rangle|^2$
    \STATE $C^{(l)} \leftarrow C^{(l)} + f(\bm{z}_m, \bm{y}^{(l)})$
    \ENDFOR
    \STATE $C^{(l)} \leftarrow C^{(l)} / N_{\text{shots}}$
    \STATE $\mathcal{L} \leftarrow \mathcal{L} + C^{(l)}$
    \ENDFOR
    \STATE $\mathcal{L} \leftarrow \mathcal{L} / N$
    \STATE Compute $\nabla_{\bm{\phi}} \mathcal{L}$ via quantum-native gradient calculation
    \STATE $\bm{\phi} \leftarrow \textsc{Optimizer}(\bm{\phi}, \nabla_{\bm{\phi}} \mathcal{L})$
    \ENDFOR
    \STATE \textbf{return} $\bm{\phi}^* \leftarrow \bm{\phi}$
  \end{algorithmic}
\end{algorithm}

\begin{algorithm}
  \caption{Test-Time Process for QEL}
  \label{alg:QEL_test}
  \begin{algorithmic}
    \REQUIRE Trained parameters $\bm{\phi}^*=\{\bm{\theta}_I^*,\bm{\theta}_F^*,\bm{w}^*\}$, test covariates $\bm{x}^{\text{test}}$, number of alternations $p$, number of shots $N_{\text{shots}}$
    \ENSURE Candidate solution $\bm{z}^*$
    \FOR{each edge $(i,j) \in E$}
    \STATE $\hat{y}_{ij} \leftarrow g_{\bm{w}^*}(\bm{x}^{\text{test}}_{ij})$
    \ENDFOR
    \STATE $\hat{\mathcal{H}}_F \leftarrow \sum_{i\in V}h_i^{\text{test}}Z_i+\sum_{(i,j)\in E} \hat{y}_{ij} Z_i Z_j$
    \STATE $|\psi\rangle \leftarrow |\bm{s}\rangle$
    \FOR{$k=1$ \textbf{to} $p$}
    \STATE $|\psi\rangle \leftarrow \exp(-i{\theta_F^{k,*}} \hat{\mathcal{H}}_F)|\psi\rangle$
    \STATE $|\psi\rangle \leftarrow \exp(-i{\theta_I^{k,*}} \mathcal{H}_I)|\psi\rangle$
    \ENDFOR
    \STATE Initialize frequency counter $\mathcal{F}$
    \FOR{$m=1$ \textbf{to} $N_{\text{shots}}$}
    \STATE $\bm{z}_m \sim |\langle \bm{z} | \psi \rangle|^2$
    \STATE $\mathcal{F}[\bm{z}_m] \leftarrow \mathcal{F}[\bm{z}_m] + 1$
    \ENDFOR
    \STATE $\bm{z}^* \leftarrow \arg\max_{\bm{z}} \mathcal{F}[\bm{z}]$
    \STATE \textbf{return} $\bm{z}^*$
  \end{algorithmic}
\end{algorithm}

\section{Proof of the Stationarity Convergence Guarantee for Joint QEL Training}
\label{app:stationary_convergence_qel}

In Theorem~\ref{prop:QEL_erm}, the objective in~\eqref{eq:QEL_loss} is shown to coincide with the empirical downstream task loss in~\eqref{eq:QEL_erm}.
Theorem~\ref{thm:qel_diminishing_step_stationarity} states a convergence-to-stationarity guarantee for the joint optimization of the QAOA parameters $\bm{\theta}$ and the contextual-encoder parameters $\bm{w}$ under a diminishing-step stochastic-gradient recursion driven by full-batch parameter-shift gradient estimates with a fixed measurement count.
This appendix introduces the technical setup, states the variational-circuit lemmas, and proves the convergence results in Theorem~\ref{thm:qel_diminishing_step_stationarity}.

\subsection{Objective and estimator.}
For each training pair $(\bm{x}^{(l)},\bm{y}^{(l)})$, define the per-instance QEL loss
\begin{equation}
\ell_l(\bm{\phi})
:=
\langle\psi(\bm{\phi},\bm{x}^{(l)})|\mathcal{H}_F(\bm{y}^{(l)})|\psi(\bm{\phi},\bm{x}^{(l)})\rangle.
\label{eq:appendixF_per_instance_loss}
\end{equation}
The QEL objective \eqref{eq:QEL_loss} is then represented as:
\begin{equation}
L(\bm{\phi})
=
\frac{1}{N}\sum_{l=1}^N \ell_l(\bm{\phi}).
\label{eq:appendixF_empirical_loss}
\end{equation}
By Theorem~\ref{prop:QEL_erm}, $L(\bm{\phi})$ is exactly the empirical downstream task loss of the stochastic policy induced by the final state in~\eqref{eq:QEL_final}. 
To analyze joint gradients, we use the gate decomposition in Appendix~\ref{apd:parametrization}. For each training instance $l$, the QEL ansatz can be written as a finite product of Pauli strings,
\begin{equation}
\mathcal{U}(\bm{\phi},\bm{x}^{(l)})
=
\prod_{q=1}^{Q}
\exp\!
\bigl[
-i\,\alpha_q^{(l)}(\bm{\phi}) P_q
\bigr],
\label{eq:appendixF_gate_decomposition}
\end{equation}
where each $P_q$ is a Pauli string satisfying $P_q^2=I$, and the effective angles $\alpha_q^{(l)}(\bm{\phi})$ collect the dependence on the QAOA parameters and on the contextual encoder through $g_{\bm{w}}(\bm{x}^{(l)}_{ij})$. This representation covers both the original parametrization in~\eqref{eq:QEL_ansatz} and the with-bias parametrization described in Appendix~\ref{apd:parametrization}.

For later use, define the angle-space loss
\begin{equation}
\widetilde{\ell}_l(\bm{a})
:=
\left\langle \bm{s} \middle|
\left(\prod_{q=1}^{Q} e^{-i a_q P_q}\right)^\dagger
\mathcal{H}_F(\bm{y}^{(l)})
\left(\prod_{q=1}^{Q} e^{-i a_q P_q}\right)
\middle| \bm{s} \right\rangle,
\qquad
\bm{a}=(a_1,\dots,a_Q)\in\mathbb{R}^Q,
\label{eq:appendixF_angle_space_loss}
\end{equation}
so that
\begin{equation}
\ell_l(\bm{\phi})
=
\widetilde{\ell}_l\bigl(\bm{\alpha}^{(l)}(\bm{\phi})\bigr),
\qquad
\bm{\alpha}^{(l)}(\bm{\phi})
:=
\bigl(\alpha_1^{(l)}(\bm{\phi}),\dots,\alpha_Q^{(l)}(\bm{\phi})\bigr).
\label{eq:appendixF_loss_def}
\end{equation}
Thus each per-instance loss $\ell_l$ is the composition of the angle-space variational objective $\widetilde{\ell}_l$ with the effective-angle map $\bm{\alpha}^{(l)}$.
For each $l$ and $q$, let $\bm{e}_q\in\mathbb{R}^Q$ denote the $q$th standard basis vector and define the shifted losses
\begin{equation}
\ell_{l,q}^{\pm}(\bm{\phi})
:=
\widetilde{\ell}_l\bigl(\bm{\alpha}^{(l)}(\bm{\phi}) \pm \tfrac{\pi}{4}\bm{e}_q\bigr).
\label{eq:appendixF_shifted_losses}
\end{equation}
At iteration \(t\), let \(\widehat{\ell}_{l,q}^{\pm}(\bm{\phi}_t)\) be the averages over \(M\) measurements of the loss evaluated on the shifted circuits obtained by replacing only the \(q\)-th effective gate angle \(\alpha_q^{(l)}(\bm{\phi}_t)\) with \(\alpha_q^{(l)}(\bm{\phi}_t)\pm \frac{\pi}{4}\), while keeping all other effective angles fixed.
Define the scalar shifted-difference estimator for the \(q\)th effective gate by
\begin{equation}
\widehat D_{l,q}(\bm{\phi}_t)
:=
\widehat{\ell}_{l,q}^{+}(\bm{\phi}_t)-\widehat{\ell}_{l,q}^{-}(\bm{\phi}_t),
\label{eq:appendixF_ps_estimator_definition}
\end{equation}
The full-batch finite-measurement gradient estimator is then
\begin{equation}
\bm{g}_t
:=
\frac{1}{N}\sum_{l=1}^{N}\sum_{q=1}^{Q}
\widehat D_{l,q}(\bm{\phi}_t)\nabla \alpha_q^{(l)}(\bm{\phi}_t).
\label{eq:appendixF_gradient_estimator_definition}
\end{equation}
The stochastic-gradient recursion analyzed below is
\begin{equation}
\bm{\phi}_{t+1} = \bm{\phi}_t - \eta_t \bm{g}_t.
\label{eq:appendixF_update_rule}
\end{equation}

\subsection{Regularity conditions.}
We work under the following conditions.

\noindent\textbf{(R1) Compact containment.}
There exists a compact convex set $\mathcal{K} \subset \mathbb{R}^{d_{\bm{\theta}}+d_{\bm{w}}}$ such that the iterates $\bm{\phi}_t=(\bm{\theta}_t,\bm{w}_t)$ remain in $\mathcal{K}$ almost surely for all $t$.

\medskip
\noindent\textbf{(R2) Effective-angle regularity.}
Each effective angle $\alpha_q^{(l)}$ defined through the Pauli-string gate decomposition in~\eqref{eq:appendixF_gate_decomposition} is $C^2$ on an open neighborhood of $\mathcal K$ for all $l=1,\dots,N$ and $q=1,\dots,Q$.
In particular, each $\alpha_q^{(l)}$ is continuously differentiable on that neighborhood.

\medskip
\noindent\textbf{(R3) Boundedness.}
For the final Hamiltonians $\mathcal{H}_F(\bm{y}^{(l)})$ defined in Section~\ref{sec:processes}, there exists $H_{\max}>0$ such that
\[
\|\mathcal{H}_F(\bm{y}^{(l)})\|_{\mathrm{op}} \le H_{\max}
\qquad \text{for all } l=1,\dots,N.
\]

\medskip
\noindent\textbf{(R4) Conditional sampling unbiasedness.}
Let $\mathcal{F}_t$ denote the sigma-algebra generated by the random initialization and all measurement outcomes revealed up to iteration $t-1$.
Conditional on $\mathcal{F}_t$, each shifted average $\widehat{\ell}_{l,q}^{\pm}(\bm{\phi}_t)$ is formed from a finite number $M\in\mathbb{N}$ of conditionally independent and identically distributed measurement outcomes with conditional mean $\ell_{l,q}^{\pm}(\bm{\phi}_t)$.

We note that \textup{(R3)} implies that every computational-basis cost is bounded:
\begin{equation}
|f(\bm{z},\bm{y}^{(l)})|
=
|\langle\bm{z}|\mathcal{H}_F(\bm{y}^{(l)})|\bm{z}\rangle|
\le
H_{\max}.
\label{eq:appendixF_measurement_outcome_bound}
\end{equation}
Finally, since each $P_q$ in~\eqref{eq:appendixF_gate_decomposition} is a Pauli string, $P_q=P_q^\dagger$, $P_q^2=I$, and $\|P_q\|_{\mathrm{op}}=1$, the two-point parameter-shift identity applies to each effective gate.

\subsection{Lemmas}
The next lemmas collect the smoothness, effective-angle, parameter-shift, and finite-measurement facts used in the proof.
To keep the appendix concise, we state them in our notation and give the specialization arguments relevant to our model, while referring to the standard quantum-gradient literature for the parameter-shift identities.

\begin{lemma}[Bounded effective-angle gradients]
\label{lem:appendixF_effective_angle_gradient_bound}
Under \textup{(R1)}--\textup{(R2)}, there exists a constant $A_1>0$ such that
\[
\|\nabla \alpha_q^{(l)}(\bm{\phi})\|_2 \le A_1
\qquad
\text{for all } \bm{\phi}\in\mathcal{K},\ l=1,\dots,N,\ q=1,\dots,Q.
\]
\end{lemma}

\begin{proof}
By \textup{(R2)}, each $\nabla\alpha_q^{(l)}$ is continuous on an open neighborhood of $\mathcal{K}$.
Since there are finitely many pairs $(l,q)$, the quantity
\[
A_1
:=
1+
\max_{1\le l\le N,\ 1\le q\le Q}
\sup_{\bm{\phi}\in\mathcal{K}}
\|\nabla\alpha_q^{(l)}(\bm{\phi})\|_2
\]
is finite, which gives the claimed uniform bound.
\end{proof}

Lemma~\ref{lem:appendixF_effective_angle_gradient_bound} records the compactness bound on the effective-angle gradients used in the finite-measurement error analysis.

\begin{lemma}[$\beta$-smoothness]
\label{lem:appendixF_smoothness}
Under the Pauli-string gate decomposition in~\eqref{eq:appendixF_gate_decomposition} and \textup{(R1)}--\textup{(R3)}, there exists a constant $\beta>0$ such that
\[
\|\nabla L(\bm{\phi})-\nabla L(\bm{v})\|_2
\le
\beta\|\bm{\phi}-\bm{v}\|_2
\qquad \text{for all } \bm{\phi},\bm{v}\in\mathcal{K}.
\]
\end{lemma}

\begin{proof}
By~\eqref{eq:appendixF_loss_def},
\[
    \ell_l(\bm{\phi})
    =
    \widetilde{\ell}_l\bigl(\bm{\alpha}^{(l)}(\bm{\phi})\bigr).
\]
Since each gate generator is a Pauli string and \textup{(R3)} bounds the corresponding Hamiltonian observable, Theorem~14 of~\citet{liu2025stochastic} gives the required angle-space smoothness bounds for $\widetilde{\ell}_l$.
By \textup{(R2)}, $\bm{\alpha}^{(l)}$ is $C^2$ on an open neighborhood of $\mathcal{K}$, so its first and second derivatives are bounded on $\mathcal{K}$.
The multivariate chain rule therefore transfers the angle-space smoothness of $\widetilde{\ell}_l$ to smoothness of $\ell_l$ with respect to the joint parameter vector $\bm{\phi}=(\bm{\theta},\bm{w})$ on $\mathcal{K}$.
Since $L$ is the finite average of the per-instance losses, the same property holds.
\end{proof}

Because Theorem~\ref{prop:QEL_erm} identifies $L(\bm{\theta},\bm{w})$ with the empirical downstream task loss, Lemma~\ref{lem:appendixF_smoothness} is equivalently a smoothness statement for the empirical-risk objective optimized by joint QEL training.

\begin{lemma}[Exact joint gradients and finite-measurement estimator]
\label{lem:appendixF_unbiased_estimator}
Under the Pauli-string gate decomposition in~\eqref{eq:appendixF_gate_decomposition} and \textup{(R1)}--\textup{(R4)}, the empirical task loss gradient admits the decomposition
\begin{equation}
\nabla L(\bm{\phi})
=
\frac{1}{N}\sum_{l=1}^{N}
J_{\bm{\alpha}^{(l)}}(\bm{\phi})^{\top}
\nabla \widetilde{\ell}_l\bigl(\bm{\alpha}^{(l)}(\bm{\phi})\bigr)
=
\frac{1}{N}\sum_{l=1}^{N}\sum_{q=1}^{Q}
\frac{\partial \widetilde{\ell}_l}{\partial a_q}
\bigl(\bm{\alpha}^{(l)}(\bm{\phi})\bigr)
\nabla \alpha_q^{(l)}(\bm{\phi}),
\label{eq:appendixF_chain_rule}
\end{equation}
where $J_{\bm{\alpha}^{(l)}}(\bm{\phi})$ denotes the Jacobian of $\bm{\alpha}^{(l)}$. For each $l$ and $q$,
\begin{equation}
\frac{\partial \widetilde{\ell}_l}{\partial a_q}
\bigl(\bm{\alpha}^{(l)}(\bm{\phi})\bigr)
=
\ell_{l,q}^{+}(\bm{\phi})-\ell_{l,q}^{-}(\bm{\phi}).
\label{eq:appendixF_parameter_shift_exact}
\end{equation}
Consequently,
\[
\nabla L(\bm{\phi})
=
\frac{1}{N}\sum_{l=1}^{N}\sum_{q=1}^{Q}
\bigl(\ell_{l,q}^{+}(\bm{\phi})-\ell_{l,q}^{-}(\bm{\phi})\bigr)
\nabla \alpha_q^{(l)}(\bm{\phi}).
\]
Moreover, the estimator in~\eqref{eq:appendixF_gradient_estimator_definition} satisfies
\begin{equation}
\mathbb E[\bm{g}_t \mid \mathcal{F}_t] = \nabla L(\bm{\phi}_t)
\label{eq:appendixF_conditional_unbiasedness}
\end{equation}
Furthermore, under the $A_1$-bound supplied by Lemma~\ref{lem:appendixF_effective_angle_gradient_bound}, there exists a constant $\sigma^2>0$ such that
\begin{equation}
\mathbb E\!\left[\|\bm{g}_t-\nabla L(\bm{\phi}_t)\|_2^2 \mid \mathcal{F}_t\right]
\le
\frac{\sigma^2}{M}
\qquad \text{for all } t.
\label{eq:appendixF_conditional_second_moment_bound}
\end{equation}
\end{lemma}

\begin{proof}
Equation~\eqref{eq:appendixF_parameter_shift_exact} is the standard exact parameter-shift identity for single-parameter gates with two-point spectrum $\{\pm 1\}$; see~\citet{schuld2019evaluating,wierichs2022general}. Applying the multivariate chain rule to
\[
\ell_l(\bm{\phi})=\widetilde{\ell}_l\bigl(\bm{\alpha}^{(l)}(\bm{\phi})\bigr)
\]
gives~\eqref{eq:appendixF_chain_rule}.
Conditioning on $\mathcal F_t$ fixes $\bm{\phi}_t$ and the factors $\nabla\alpha_q^{(l)}(\bm{\phi}_t)$; by \textup{(R4)}, the only remaining randomness is the fresh measurement randomness.

By \textup{(R4)}, the scalar estimator in~\eqref{eq:appendixF_ps_estimator_definition} satisfies
\[
\mathbb E[\widehat D_{l,q}(\bm{\phi}_t)\mid\mathcal F_t]
=
\ell_{l,q}^{+}(\bm{\phi}_t)-\ell_{l,q}^{-}(\bm{\phi}_t).
\]
This is the finite-measurement unbiasedness used in parameter-shift estimators; see \citet{sweke2020stochastic,mari2021estimating}.
Substituting this identity into~\eqref{eq:appendixF_gradient_estimator_definition} and using linearity of conditional expectation and~\eqref{eq:appendixF_parameter_shift_exact} gives~\eqref{eq:appendixF_conditional_unbiasedness}.

Moreover, by \eqref{eq:appendixF_measurement_outcome_bound}, each shifted average over \(M\) measurements has conditional variance at most $\frac{H_{\max}^2}{M}$.
Therefore,
\[
\mathbb E\!\left[
\left(
\widehat D_{l,q}(\bm{\phi}_t)
-
\bigl(\ell_{l,q}^{+}(\bm{\phi}_t)-\ell_{l,q}^{-}(\bm{\phi}_t)\bigr)
\right)^2
\middle|\mathcal F_t
\right]
\le
\frac{4H_{\max}^2}{M}.
\]
Together with the bound $\|\nabla\alpha_q^{(l)}(\bm{\phi}_t)\|_2\le A_1$ from Lemma~\ref{lem:appendixF_effective_angle_gradient_bound}, this yields
\[
\begin{aligned}
\mathbb E\!\left[\|\bm g_t-\nabla L(\bm{\phi}_t)\|_2^2\mid\mathcal F_t\right] \le
\frac{1}{N^2}
\left(
\sum_{l=1}^N\sum_{q=1}^Q
\|\nabla\alpha_q^{(l)}(\bm{\phi}_t)\|_2
\frac{2H_{\max}}{\sqrt M}
\right)^2
\le
\frac{4Q^2A_1^2H_{\max}^2}{M},
\end{aligned}
\]
so~\eqref{eq:appendixF_conditional_second_moment_bound} holds with, for instance, $\sigma^2=4Q^2A_1^2H_{\max}^2$.
\end{proof}

Lemma~\ref{lem:appendixF_unbiased_estimator} separates the deterministic chain-rule and parameter-shift identities from the finite-measurement stochastic error.
It is the step that connects the implementable gradient estimator in~\eqref{eq:appendixF_gradient_estimator_definition} to the exact gradient used in the stochastic-approximation proof below.

\subsection{Proof of Theorem~\ref{thm:qel_diminishing_step_stationarity}}
\begin{proof}
Let
\[
\bm{\phi}_t := (\bm{\theta}_t,\bm{w}_t),
\qquad
L_t := L(\bm{\phi}_t),
\qquad
\nabla_t := \nabla L(\bm{\phi}_t),
\]
and define
\[
\bar E_0
:=
\frac{1}{N}\sum_{l=1}^N
\lambda_{\min}\!\bigl(\mathcal{H}_F(\bm{y}^{(l)})\bigr).
\]
For each training instance $l$,
\[
\ell_l(\bm{\phi}_t)
=
\langle\psi(\bm{\phi}_t,\bm{x}^{(l)})|
\mathcal{H}_F(\bm{y}^{(l)})
|\psi(\bm{\phi}_t,\bm{x}^{(l)})\rangle.
\]
Since each \(\mathcal{H}_F(\bm{y}^{(l)})\) is Hermitian by construction, the variational principle in~\eqref{eq:variational} gives
\[
    \ell_l(\bm{\phi})
    \ge
    \lambda_{\min}\!\bigl(\mathcal{H}_F(\bm{y}^{(l)})\bigr),
    \qquad \forall \bm{\phi}.
\]
Applying this bound to the stochastic iterate \(\bm{\phi}_t\) and averaging over \(l\) gives
\[
    L_t=
    \frac{1}{N}\sum_{l=1}^N \ell_l(\bm{\phi}_t)
    \ge
    \bar E_0.
\]
Thus $Y_t:=L_t-\bar E_0\ge0$ almost surely for every $t$.

By Lemma~\ref{lem:appendixF_smoothness}, $L$ is $\beta$-smooth on $\mathcal{K}$. Hence, for all iterates in $\mathcal K$, the descent lemma gives
\[
L_{t+1}
\le
L_t
+
\langle \nabla_t,\bm{\phi}_{t+1}-\bm{\phi}_t\rangle
+
\frac{\beta}{2}\|\bm{\phi}_{t+1}-\bm{\phi}_t\|_2^2.
\]
Substituting~\eqref{eq:appendixF_update_rule} yields
\[
L_{t+1}
\le
L_t
-
\eta_t\langle \nabla_t,\bm{g}_t\rangle
+
\frac{\beta\eta_t^2}{2}\|\bm{g}_t\|_2^2.
\]
Taking conditional expectation with respect to $\mathcal{F}_t$ and using~\eqref{eq:appendixF_conditional_unbiasedness},
\[
\mathbb E[L_{t+1}\mid \mathcal{F}_t]
\le
L_t
-
\eta_t\|\nabla_t\|_2^2
+
\frac{\beta\eta_t^2}{2}\mathbb E[\|\bm{g}_t\|_2^2\mid \mathcal{F}_t].
\]
Writing $\bm{g}_t=\nabla_t+(\bm{g}_t-\nabla_t)$, the cross term vanishes under conditional expectation by~\eqref{eq:appendixF_conditional_unbiasedness}. Therefore,
\[
\mathbb E[\|\bm{g}_t\|_2^2\mid \mathcal{F}_t]
=
\|\nabla_t\|_2^2
+
\mathbb E[\|\bm{g}_t-\nabla_t\|_2^2\mid \mathcal{F}_t].
\]
Applying~\eqref{eq:appendixF_conditional_second_moment_bound} gives
\[
\mathbb E[L_{t+1}\mid \mathcal{F}_t]
\le
L_t
-
\eta_t\Bigl(1-\frac{\beta\eta_t}{2}\Bigr)\|\nabla_t\|_2^2
+
\frac{\beta\eta_t^2\sigma^2}{2M}.
\]
By the step-size assumption in Theorem~\ref{thm:qel_diminishing_step_stationarity}, $1-\frac{\beta\eta_t}{2}\ge \frac{1}{2}$, hence
\begin{equation}
\mathbb E[L_{t+1}\mid \mathcal{F}_t]
\le
L_t
-
\frac{\eta_t}{2}\|\nabla_t\|_2^2
+
\frac{\beta\eta_t^2\sigma^2}{2M}.
\label{eq:appendixF_diminishing_conditional_descent}
\end{equation}

Since the shifted averages are bounded by $H_{\max}$ and Lemma~\ref{lem:appendixF_effective_angle_gradient_bound} gives $\|\nabla \alpha_q^{(l)}(\bm{\phi})\|_2\le A_1$ on $\mathcal{K}$, the estimator is almost surely bounded,
\[
\|\bm{g}_t\|_2
\le
\frac{1}{N}\sum_{l=1}^{N}\sum_{q=1}^{Q}
|\widehat D_{l,q}(\bm{\phi}_t)|\,\|\nabla\alpha_q^{(l)}(\bm{\phi}_t)\|_2
\le
2QH_{\max}A_1
=:
G_{\max}.
\]
If $G_{\max}=0$, then $\bm g_t=0$ almost surely for every $t$, and conditional unbiasedness implies $\nabla_t=0$ almost surely for every $t$. The claim is immediate in this degenerate case, so assume $G_{\max}>0$.

From~\eqref{eq:appendixF_diminishing_conditional_descent},
\[
\mathbb E[Y_{t+1}\mid\mathcal F_t]
\le
Y_t
-
\frac{\eta_t}{2}\|\nabla_t\|_2^2
+
\frac{\beta\eta_t^2\sigma^2}{2M}.
\]
Since \(Y_t\ge0\) and \(\sum_t \frac{\beta\eta_t^2\sigma^2}{2M}<\infty\), the Robbins--Siegmund supermartingale convergence theorem \citep{robbins1971convergence} implies
\begin{equation}
\sum_{t=0}^{\infty}\eta_t\|\nabla_t\|_2^2<\infty
\qquad\text{almost surely}.
\label{eq:appendixF_weighted_gradient_summable_as}
\end{equation}
On the same probability-one event, suppose for contradiction that $\|\nabla_t\|_2$ does not converge to zero.
Then there exists $\varepsilon>0$ and infinitely many indices at which $\|\nabla_t\|_2\ge 2\varepsilon$. Since $\sum_t\eta_t=\infty$, relation~\eqref{eq:appendixF_weighted_gradient_summable_as} also implies that $\|\nabla_t\|_2<\varepsilon$ infinitely often.
Thus there are infinitely many disjoint excursions that start at an index $r$ with $\|\nabla_r\|_2\ge 2\varepsilon$ and end at the first subsequent index $s>r$ with $\|\nabla_s\|_2<\varepsilon$.

For each such excursion, smoothness of $L$ and the estimator bound give
\[
\bigl|\|\nabla_{t+1}\|_2-\|\nabla_t\|_2\bigr|
\le
\|\nabla_{t+1}-\nabla_t\|_2
\le
\beta\|\bm{\phi}_{t+1}-\bm{\phi}_t\|_2
\le
\beta G_{\max}\eta_t.
\]
Therefore the cumulative step length during the excursion must satisfy
\[
\sum_{t=r}^{s-1}\eta_t
\ge
\frac{\varepsilon}{\beta G_{\max}}.
\]
Moreover, by the definition of $s$, $\|\nabla_t\|_2\ge\varepsilon$ for every $t=r,\ldots,s-1$. Hence the excursion contributes at least
\[
\sum_{t=r}^{s-1}\eta_t\|\nabla_t\|_2^2
\ge
\varepsilon^2\sum_{t=r}^{s-1}\eta_t
\ge
\frac{\varepsilon^3}{\beta G_{\max}}
\]
to the series in~\eqref{eq:appendixF_weighted_gradient_summable_as}.
Infinitely many disjoint excursions would force that series to diverge, contradicting~\eqref{eq:appendixF_weighted_gradient_summable_as}.
Therefore, in the almost-sure sense,
\[
    \mathbb P\!\left(
        \lim_{t\to\infty}\|\nabla L(\bm{\phi}_t)\|_2 = 0
    \right)=1.
\]

By compact containment in \textup{(R1)} and continuity of $\nabla L$ on $\mathcal K$, there exists $C_\nabla<\infty$ such that
\[
\|\nabla L(\bm{\phi})\|_2^2\le C_\nabla
\qquad\text{for all } \bm{\phi}\in\mathcal K.
\]
Since $\bm{\phi}_t\in\mathcal K$ almost surely, the sequence $\|\nabla L(\bm{\phi}_t)\|_2^2$ is dominated by $C_\nabla$. The almost-sure convergence established above and the dominated convergence theorem imply
\[
\lim_{t\to\infty}
\mathbb E\!\left[\|\nabla L(\bm{\phi}_t)\|_2^2\right]
=0.
\]
This completes the proof.
\end{proof}

\begin{remark}[Interpretation and scope]
Theorem~\ref{thm:qel_diminishing_step_stationarity} is a convergence-to-stationarity guarantee for the joint optimization of $\bm{\theta}$ and $\bm{w}$ under the stochastic gradient recursion and the parameter-shift rule with a fixed finite number of measurements per shifted circuit and the regularity conditions, including the compact containment assumption (R1).
It does not assert convergence of the parameters themselves, nor convergence to a global minimizer of the full nonconvex landscape.
Instead, the theorem proves vanishing expected squared gradient norm and almost-sure stationarity of the full stochastic trajectory under the regularity conditions,
\[
\lim_{t\to\infty}
\mathbb E\!\left[\|\nabla L(\bm{\phi}_t)\|_2^2\right]
=0,
\qquad
\mathbb P\!\left(
            \lim_{t\to\infty}\|\nabla L(\bm{\phi}_t)\|_2 = 0
        \right)=1.
\]
Since Theorem~\ref{prop:QEL_erm} identifies $L(\bm{\phi})$ with the empirical downstream task loss, this guarantee applies to the empirical expected task loss of the stochastic policy optimized by QEL, rather than to a surrogate prediction loss.
\end{remark}

\section{Implementation Details}
\label{apd:implementation}

\subsection{Hardware and Software}
All numerical experiments are performed on a machine with 256GB RAM, Intel Xeon Gold 5317 CPUs (24 cores total), and 4 NVIDIA GeForce RTX 4090 GPUs with Python 3.12.12.
Quantum circuits are constructed using PennyLane (version 0.43.1) \citep{bergholm2018pennylane}, with the PennyLane Lightning GPU simulator (version 0.43.0).
To calculate the true optimal solutions for measuring decision performance, Gurobi (version 12.0.3) \citep{gurobi} is used as the optimization solver for MaxCut and QAP.
For BMP, we use CVXPY (version 1.8.2) \citep{diamond2016cvxpy, agrawal2018rewriting} to follow \citep{geng2024benchmarking}.
For List-LTR and LODL, the corresponding optimization solver is also repeatedly invoked during training.
List-LTR uses it to initialize and update the solution cache, while LODL uses it to obtain optimal decisions for sampled perturbations when learning surrogate losses.
Classical benchmarks are implemented using PyTorch (version 2.9.1) \cite{paszke2019pytorch} in accordance with the original study.

\subsection{Data Generation}
For MaxCut and QAP, we generate synthetic datasets with $N=512$ instances.
For MaxCut, covariates $\bm{x}_{ij} \in \mathbb{R}^2$ are sampled uniformly from $[-2.048, 2.048]^2$ for each edge $(i,j) \in E$.
The uncertain edge weights $y_{ij}$ are generated by applying the Rosenbrock function \citep{rosenbrock1960automatic} to the covariates, followed by a log-transformation $\log(1+\cdot)$ and additive Gaussian noise with standard deviation equal to 10\% of the signal standard deviation.
For QAP, covariates $\bm{x}_{ij} \in \mathbb{R}^2$ are sampled uniformly from $[-2, 2]^2$ for each facility-location pair.
The uncertain flow matrix entries $y_{ij}$ are generated using the Goldstein-Price function \citep{goldsteinprice} with the same log-transformation and noise injection procedure, then symmetrized.
The distance matrix $D$ is sampled uniformly from $[0, 1]$ and symmetrized independently for each instance.
These benchmark functions are chosen for their nonlinear, multimodal characteristics that create challenging optimization tasks while maintaining reproducibility.
For BMP, we follow the benchmark design in \citep{geng2024benchmarking}, constructing each instance from the raw Cora citation graph by sampling two disjoint node sets of equal size and forming all cross-partition candidate edges.
The Cora dataset is distributed via ICPSR/openICPSR under the Creative Commons Attribution 4.0 International (CC BY 4.0) License, and our experiments use a sampled and preprocessed version of the citation graph.
Specifically, due to the memory limitations of quantum simulation, we restrict the sizes of the two disjoint node sets to 5.
We also reduce the original 1,433 bag-of-words node features of Cora to 128 by selecting the most frequently activated features across all nodes.

\subsection{Hyperparameters}
For MaxCut and QAP, we adopt the train-validation-test split, partitioned into 384 for training, 64 for validation, and 64 for testing.
For BMP, we increase the sizes of the training, validation, and test sets by a factor of 5 compared with \citet{geng2024benchmarking}.
For QEL experiments, we use a learning rate of 0.001, batch size of 8, the number of shots of 4,096 for test, and train for up to 30 epochs with early stopping patience of 10.
Note that the adjoint method of PennyLane does not require sampling shots.
For penalization, in QAP, penalty values of 50 and 150 are employed for 16- and 25-variable instances, respectively.
For BM, we use a penalty value of 1.5 and batch size of 2.
For classical benchmarks, most hyperparameters are set to match those in \citep{geng2024benchmarking}.
We decrease the number of perturbed samples for LODL to 500 to avoid iteratively solving too many NP-hard MaxCut and QAP problems.
We use the Adam optimizer \cite{kingma2014adam}.
Although Theorem~\ref{thm:qel_diminishing_step_stationarity} establishes stationarity only for a diminishing-step-size SGD recursion, we follow \cite{geng2024benchmarking} and use the same Adam optimizer for QEL and the classical benchmark models to ensure a consistent empirical comparison.

\section{Problem Formulations}
\label{apd:problem}
\subsection{Undirected Weighted Maximum Cut Problem}
Given an undirected graph $G=(V,E)$ with edge weights $w_{ij}$ for each $(i,j) \in E$, the undirected weighted maximum cut problem seeks a partition of the vertices into two disjoint sets such that the total weight of edges crossing the partition is maximized. It is formulated as
\begin{align*}
    \max\quad&\sum_{(i,j)\in E}w_{ij}(x_i+x_j-2x_ix_j),\\ 
    \textrm{s.t.}\quad&x_i\in\{0,1\}\quad \forall i\in V.
\end{align*}
It is represented as the Ising Hamiltonian to be minimized by mapping $x_i\mapsto\frac{1}{2}(1-Z_i)$,
\begin{align*}
    \min\quad&\sum_{(i,j)\in E}-\frac{w_{ij}}{2}(1-Z_iZ_j),\\
    \textrm{s.t.}\quad&Z_i\in\{-1,1\}\quad \forall i\in V.
\end{align*}
We set weights $w_{ij}$ as uncertain coefficients to be estimated.

\subsection{Quadratic Assignment Problem}
We are given $n$ facilities and $n$ locations, along with the flow $f_{ij}$ between facilities $i$ and $j$ and the distance $d_{kl}$ between locations $k$ and $l$.
Our goal is to find an assignment of facilities to locations that minimizes the cost defined by the flow and the distance. It is formulated as
\begin{align*}
    \min\quad&\sum_{i,j=1}^n\sum_{k,l=1}^nf_{ij}d_{kl}x_{ik}x_{jl},\\
    \textrm{s.t.}\quad&\sum_{i=1}^nx_{ik}=1\quad\forall k\in\{1,\cdots,n\} \\
                      &\sum_{k=1}^nx_{ik}=1\quad\forall i\in\{1,\cdots,n\} \\
    &x_{ik}\in\{0,1\}\quad \forall i,k\in\{1,\cdots,n\}.
\end{align*}
To transform it into the terms of the Ising Hamiltonian, we penalize equality constraints with a fixed penalty value $P$,
\begin{align*}
    \min\quad&\sum_{i,j=1}^n\sum_{k,l=1}^n\frac{f_{ij}d_{kl}}{4}(1-Z_{ik})(1-Z_{jl})\\&+P\sum_{k=1}^n\bigl(1-\sum_{i=1}^n\frac{1}{2}(1-Z_{ik})\bigr)^2+P\sum_{i=1}^n\bigl(1-\sum_{k=1}^n\frac{1}{2}(1-Z_{ik})\bigr)^2,\\
    \textrm{s.t.}\quad&Z_{ik}\in\{-1,1\}\quad \forall i,k\in\{1,\cdots,n\}.
\end{align*}
We set flows $f_{ij}$ as uncertain coefficients to be estimated.

\subsection{Bipartite Matching Problem}
Given two disjoint node sets $U$ and $W$ with $|U|=|W|=n$, and edge profits $p_{ij}$ for each $(i,j)\in U\times W$, the bipartite matching problem seeks a perfect matching that maximizes the total profit. It is formulated as
\begin{align*}
    \max\quad&\sum_{i=1}^n\sum_{j=1}^np_{ij}x_{ij},\\
    \textrm{s.t.}\quad&\sum_{j=1}^nx_{ij}=1\quad\forall i\in\{1,\cdots,n\}, \\
                      &\sum_{i=1}^nx_{ij}=1\quad\forall j\in\{1,\cdots,n\}, \\
                      &x_{ij}\in\{0,1\}\quad \forall i,j\in\{1,\cdots,n\}.
\end{align*}
To transform it into the terms of the Ising Hamiltonian, we penalize equality constraints with a fixed penalty value $P$,
\begin{align*}
    \min\quad&-\sum_{i=1}^n\sum_{j=1}^n\frac{p_{ij}}{2}(1-Z_{ij})\\&+P\sum_{i=1}^n\bigl(1-\sum_{j=1}^n\frac{1}{2}(1-Z_{ij})\bigr)^2+P\sum_{j=1}^n\bigl(1-\sum_{i=1}^n\frac{1}{2}(1-Z_{ij})\bigr)^2,\\
    \textrm{s.t.}\quad&Z_{ij}\in\{-1,1\}\quad \forall i,j\in\{1,\cdots,n\}.
\end{align*}
We set edge profits $p_{ij}$ as uncertain coefficients to be estimated.

\section{Additional Experimental Results}
\label{apd:results}
Table~\ref{tab:confidence-intervals} reports the 95\% confidence intervals for the relative regrets in Table~\ref{tab:result}.
\begin{table}[H]
\centering
\caption{95\% confidence intervals of each relative regret in Table~\ref{tab:result}.}
\label{tab:confidence-intervals}
\begin{adjustbox}{max width=\textwidth}
\begin{tabular}{cclrrrrr}
\hline
                         &          &  & \multicolumn{5}{c}{Regret (\%)}                                                                                                 \\ \cline{4-8} 
                         &          &  & \multicolumn{2}{c}{MaxCut}                        & \multicolumn{2}{c}{QAP}                           & \multicolumn{1}{c}{BMP} \\ \cline{4-8} 
                         &          &  & \multicolumn{1}{c}{16V} & \multicolumn{1}{c}{25V} & \multicolumn{1}{c}{16V} & \multicolumn{1}{c}{25V} & \multicolumn{1}{c}{25V} \\ \hline
LODL                     & $(2,32)$ &  & (8.50, 9.98)            & (8.31, 9.45)            & (3.25, 4.83)            & (4.35, 5.74)            & (80.00, 80.00)          \\ \hline
List-LTR                 & $(2,32)$ &  & (8.65, 10.03)           & (7.84, 8.77)            & (4.01, 5.95)            & (4.65, 6.30)            & (77.19, 97.21)          \\ \hline
\multirow{2}{*}{QEL-Lin} & $p=3$    &  & (4.96, 6.16)            & (4.29, 5.09)            & (3.72, 5.42)            & (5.64, 6.95)            & (65.44, 94.56)          \\
                         & $p=4$    &  & (4.85, 6.14)            & (4.53, 5.41)            & (4.67, 6.72)            & (5.08, 6.43)            & (65.44, 94.56)          \\ \hline
\multirow{2}{*}{QEL-Log} & $p=3$    &  & (4.55, 5.79)            & (4.05, 4.98)            & (3.69, 5.40)            & (4.50, 5.86)            & (59.59, 90.41)          \\
                         & $p=4$    &  & (4.51, 5.62)            & (3.81, 4.61)            & (3.67, 5.26)            & (4.84, 6.35)            & (42.17, 77.83)          \\ \hline
\end{tabular}
\end{adjustbox}
\end{table}

\newpage

\end{document}